\documentclass[conference]{IEEEtran}

\usepackage[english]{babel}
\usepackage{amsmath,amssymb}
\usepackage{times}
\usepackage{relsize}
\usepackage{graphicx}
\usepackage{url}
\usepackage{booktabs}
\usepackage{multirow}
\usepackage{mparhack}
\usepackage{subfigure}
\usepackage{authblk}
\usepackage{amsthm}
\usepackage{mathrsfs}

\usepackage[noend]{algorithmic}
\usepackage[linesnumbered,ruled,vlined]{algorithm2e}

\usepackage{array}
\usepackage{cite}
\usepackage{eqparbox}
\usepackage{mdwmath}

\usepackage{epsfig}
\usepackage{xcolor}

\usepackage{eulervm}

\usepackage{mathtools}
\usepackage{bm}
\usepackage{bbold}
\usepackage[all]{xy}
\usepackage{etoolbox}
\DeclareMathAlphabet{\mathantt}{OT1}{antt}{li}{it}
\DeclareMathAlphabet{\mathpzc}{OT1}{pzc}{m}{it}

\DeclarePairedDelimiter\norm{\lVert}{\rVert}%

\usepackage{accents}

\newtheorem{theorem}{Theorem}

\newtheorem{lemma}{Lemma}

\DeclareFontFamily{OT1}{pzc}{}
\DeclareFontShape{OT1}{pzc}{m}{it}%
  {<-> s * [1.1] pzcmi7t}{}
\DeclareMathAlphabet{\mathpzc}{OT1}{pzc}%
                     {m}{it}

\def\I{\mathcal{I}}
\def\J{\mathcal{J}}

\def\u{\mathpzc{u}}
\def\U{\mathcal{U}}
\def\P{\mathcal{P}}

\DeclareMathOperator{\argmin}{\arg\min}

\renewcommand{\vec}[1]{\mathbold{#1}}
\renewcommand{\bm}[1]{\mathbold{#1}}

\title{A Framework for Optimizing Multi-cell NOMA: \\ Delivering Demand with Less Resource}

\author[1]{Lei You}
\author[2]{Lei Lei}
\author[1]{Di Yuan}
\author[3]{Sumei Sun}
\author[2]{Symeon Chatzinotas}
\author[2]{Bj\"{o}rn Ottersten}

\affil[1]{{\small Department of Information Technology, Uppsala University,
Sweden}}
\affil[2]{{\small Interdisciplinary Centre for Security, Reliability and Trust, Luxembourg University, Luxembourg}}
\affil[3]{{\small Institute for Infocomm Research, A*STAR, Singapore}}
\affil[ ]{\em{ Emails: \{lei.you; di.yuan\}@it.uu.se, \{lei.lei; symeon.chatzinotas; bjorn.ottersten\}@uni.lu, sunsm@i2r.a-star.edu.sg}\/
\vspace{-0.0cm}
}

\begin{document}

\maketitle

\begin{abstract}
Non-orthogonal multiple access (NOMA) allows multiple users to simultaneously access the same time-frequency resource by using superposition coding and successive interference cancellation (SIC). Thus far, most papers on NOMA have focused on performance gain for one or sometimes two base stations. In this paper, we study multi-cell NOMA and provide a general framework for user clustering and power allocation, taking into account inter-cell interference, for optimizing resource allocation of NOMA in multi-cell networks of arbitrary topology. We provide a series of theoretical analysis, to algorithmically enable optimization approaches. The resulting algorithmic notion is very general. Namely, we prove that for any performance metric that monotonically increases in the cells' resource consumption, we have convergence guarantee for global optimum. We apply the framework with its algorithmic concept to a multi-cell scenario to demonstrate the gain of NOMA in achieving significantly higher efficiency.
\end{abstract}

\section{Introduction}

To what extent can non-orthogonal multiple access (NOMA) improve network resource efficiency? In two recent surveys~\cite{7676258,2016arXiv161101607S}, the authors pointed out that resource allocation in multi-cell NOMA poses more research challenges compared to the single-cell case, because optimizing NOMA with multiple cells has to model the interplay between successive interference cancellation (SIC) and inter-cell interference. As one step forward, the investigations in ~\cite{7676258,2016arXiv161101607S} have addressed two-cell scenarios. To the best of our knowledge, enhancing network resource efficiency in multi-cell NOMA with user pairing has not been addressed yet. In~\cite{7582424}, the authors proposed two interference alignment based coordinated beamforming methods in two-cell scenarios. Reference~\cite{DBLP:journals/corr/TabassumHH16} uses stochastic geometry to model the inter-cell interference in NOMA. The crucial aspect of multi-cell NOMA consists of capturing mutual influence among cells. 
In the past decade, a modeling approach that characterizes the inter-cell interference via capturing the mutual influence among the load of cells, i.e., \textit{load coupling}, had been proposed and widely adopted for orthogonal multiple access (OMA) networks~\cite{7132788,5198628,6204009,6732895,7585124,7880696}.
However, this approach does not apply to NOMA, because time-frequency resource sharing via SIC is not allowed. Whether or not the type of model in OMA can be extended to NOMA has remained open until now. 

The main contribution of this paper is that, \textit{we provide a general framework for obtaining optimal user clustering and power allocation in interference-coupled multi-cell NOMA for resource efficiency.} More specifically,
\begin{enumerate}
	\item We make a significant generalization for the interference models being used in \cite{7132788,5198628,6204009,6732895,7585124,7880696} to multi-cell NOMA. Theoretical analysis in terms of feasibility and computation are provided.
	\item Based on 1), a unified optimization framework for jointly optimizing user clustering and power allocation in multi-cell NOMA is derived, to achieve global optimum for \textit{any} performance metric that monotonically increases in the cells' resource consumption.
	\item We demonstrate the gain of NOMA in multi-cell scenarios and show NOMA is indeed a promising solution for meeting user demands with less resource than OMA.
\end{enumerate}

\section{A Multi-cell Interference Averaging Model}
\label{sec:sys_model}

\subsection{Network Model}
We consider downlink, and remark that the framework can be straightforwardly extended to uplink. Denote by $\I=\{1,2,\ldots,n\}$ and $\J=\{n+1,n+2,\ldots,n+m\}$ the sets of cells and UEs, respectively. Denote by $\J_i$ the set of UEs served by cell $i$, with $i\in\I$. Denote by $g_{ij}$ the path loss factor from cell $i$ to UE $j$, with $i\in\I$ and $j\in\J$. 
When using $j$ to refer to one UE in $\J$, $i$ by default indicates $j$'s serving cell, unless stated otherwise. 

We use $\mathpzc{u}$ as a generic notation for UE \textit{cluster} (referred to as ``cluster'' in the remaining context), i.e. a set consisting of one or multiple users that are allowed to access the same time-frequency resource by SIC. {With more users being put in a cluster, the complexity for decoding in SIC grows fast~\cite{7676258,2016arXiv161101607S}. For the sake of this practical consideration, we follow the assumption used in other references~\cite{7676258,2016arXiv161101607S,7273963,R1-153332,7582424,7357604} that up to two users are clustered together\footnote{Reference \cite{7273963} demonstrated most of the possible performance improvement can be achieved by two-users clustering in NOMA.}. If there is a need to differentiate between clusters, we put indices on $\u$, e.g., $\u_1,\u_2,\u_3,\ldots$. For UE clustering in cell $i$, denote by $\U_i$ the set of candidate clusters. We have $\U_i\cap\U_k=\phi$ for any $i\neq k$ with $i,k\in\I$. Similarly, denote by $\U_j$ ($j\in\J$) the set of clusters containing UE $j$. Let $\U=\bigcup_{i\in\I}\U_{i}$ (or equivalently $\U=\bigcup_{j\in\J}\U_{j}$) be the set of all clusters. Note that one UE may belong to multiple user clusters, e.g. $\u_1=\{1,2\}$, $\u_2=\{1,3\}$ with UE $1$ belonging to both $\u_1$ and $\u_2$. To keep the generality of our model for extreme case (e.g. there is only one UE in a cell), a cluster may consist of a single UE, i.e. $\u_1=\{1\}$ and $\u_2=\{2\}$.

The time-frequency domain resource that is divided into resource blocks (RBs). Let $p_i$ be the transmission power on any RB in cell $i$. For any cluster $\u=\{j,h\}$ in cell $i$, RB(s) can be accessed together (i.e., shared) by UEs $j$ and $h$. On any of the shared RBs, \textit{power splitting} is done on $p_i$, with $p_{j\u}$ and $p_{h\u}$ allocated to $j$ and $h$, respectively, and $p_{j\u}+p_{h\u}=p_i$. On one RB, for any UE $j$ and any cluster $\u$ ($j\in\u$), the signal-to-interference and noise ratio (SINR) is computed by:
\begin{equation}
\gamma_{j\u} = \frac{p_{j\u}g_{ij}}
{\underbrace{\sum_{\substack{h\in\u:\\ b_{\u}(h)<b_{\u}(j)}}
p_{h\u}g_{ij}}_{\textnormal{intra-cell}} 
+\underbrace{\sum_{k\in\I\backslash\{i\}}
   I_{kj}
}_{\textnormal{inter-cell}}
      +\sigma^2
}
\label{eq:sinr}	
\end{equation}

In~\eqref{eq:sinr}, $\sigma^2$ is the noise power. Parameter $I_{kj}$ refers to the interference from cell $k$ to UE $j$. 
In~\cite{tse2005fundamentals} (Chapter 6.2.2, pp. 238) it is shown that, with superposition coding, a user can decode the data of another user with worse channel gain. The user with worse channel condition is subject to intra-cell interference. 
We use bijection $b_{\u}(j)\rightarrow\{1,2\}$ ($\u\in\U$ and $j\in\u$) to represent the decoding order. Based on the bijection, the UE with value $1$ decodes the UE with value $2$, and the UE with value $2$ receives intra-cell interference from the UE with value $1$. The decoding order is not constrained by power splitting~\cite{tse2005fundamentals}, even though by our numerical results, more power is always allocated to the one with worse channel. The decoding order is fixed. The issue of the influence of inter-cell interference on the decoding order is addressed later in Section~\ref{subsec:decoding}.

\subsection{Multi-cell Interference Modeling}
\label{subsec:averaging}

Interference modeling based on considering the amount of resource consumption has been widely used for OMA. The method is specified as follows. Denote by $\rho_{k}$ the proportion of RBs that are allocated for serving UEs in cell $k$. If cell $k$ is fully loaded, meaning that all RBs are allocated, then $\rho_{k}=1$. Another extreme case is that cell $k$ is idle within the time interval in question, and accordingly $\rho_{k}=0$. For the two cases, consider any UE $j$ served by cell $i$. The exact interference $j$ receives from cell $k$ is $I_{kj}=p_{k}g_{kj}$ and $I_{kj}=0$, respectively. For the former case, cell $k$ interferes with every RB in cell $i$. For the latter, no interference is caused by cell $k$, as none of the RBs in cell $k$ are active when $\rho_k=0$. For $0<\rho_{k}<1$, a balance is stroked between exactness and simplicity by averaging on the interference within the time-frequency domain, see~\eqref{eq:Ikj}. This interference averaging technique was used in \cite{7132788,5198628,6204009,6732895,7585124,7880696}.
\begin{equation}
I_{kj} = p_{k}g_{kj}\rho_{k}	
\label{eq:Ikj}
\end{equation}

Intuitively, $\rho_{k}$ reflects the likelihood that a UE outside cell $k$ receives interference from $k$. By the definition of $\rho_k$, it can be interpreted as the \textit{load of cell} $k$, and \textit{used for measuring the time-frequency resource consumption} of cell $k$. An explanation of \eqref{eq:Ikj} is that, the inter-cell interference incurred by a cell is directly proportional to the cell's load, which has a direct correlation to the number of served UEs and the intensity of the cell's data traffic.

\subsection{Decoding}
\label{subsec:decoding}
The modeling complexity increases significantly for NOMA because inter-cell interference influences the decoding order.
The modeling task is approached by identifying those clusters of which the decoding orders are decoupled from the inter-cell interference. UEs fulfilling Lemma~\ref{rmk:decoding} below are theoretically guaranteed to be independent of the inter-cell interference in respect of decoding in SIC. The proof of Lemma~\ref{rmk:decoding} is in the Appendix. Clusters consisting of UEs violating Lemma~\ref{rmk:decoding} are excluded from $\U$ and the model complexity is thus significantly reduced.
\begin{lemma}
Suppose two users $j$ and $h$ within cluster $\u$ are served by cell $i$ ($g_{ij}>g_{ih}$).
If $g_{ij}/g_{ih}\geq g_{kj}/g_{kh}$ for all $k\in\I\backslash\{i\}$, then $b_{\u}(j)=1$ and $b_{\u}(h)=2$. 
\label{rmk:decoding}
\end{lemma}

In practical consideration, Lemma~\ref{rmk:decoding} reduces the user clustering complexity without damaging the performance. As pointed out by~\cite{7307220,6954404}, the large scale path-loss is a practically reasonable factor for ranking the decoding order. As for user clustering in NOMA,
two users with disparate channels from the cell are preferred
to be clustered for achieving good performance~\cite{7273963,7357604}.
Consider a cluster $\u=\{j,h\}$ of cell $i$. If $g_{ij}\gg g_{ih}$, then most likely $g_{ij}/g_{ih}>g_{kj}/g_{kh}$ for $k\in\I\backslash\{i\}$, as the large scale path loss from other cells, tends not to differ as much as from the serving cell $i$ in this case.

\subsection{Cell Load Computation with User Clustering}

Denote by $d_j$ the bit demand of UE $j$ with $j\in\J$. Let $B$ be the spectral bandwidth on each RB. Denote by $M$ the total number of RBs. Since the term $B\log(1+\gamma_{j\u})$ represents the capacity of one RB for UE $j$ in cluster $\u$ ($j\in\u$), the total achievable capacity on all RBs with respect to $j$ and $\u$ is computed by 
\begin{equation}
    c_{j\u}= MB
    \log\left(1+\gamma_{j\u}
\right).
\label{eq:cju}
\end{equation}
Denote by $x_{\u}$ the proportion of allocated RB(s) to cluster $\u$. The sum of $x_{\u}$ for $\u\in\U_i$ equals the load of cell $i$, as shown by~\eqref{eq:rhoi}, where $\bar{\rho}$ represents the load limit.
\begin{equation}
    \rho_{i} = \sum_{\u\in\U_i}x_{\u}\leq\bar{\rho}
\label{eq:rhoi}
\end{equation}
Note that the term $c_{j\u}x_{\u}$ computes the achieved bits for UE $j$ in cluster $\u$ with allocated proportion of time-frequency resource $x_{\u}$. To satisfy the quality-of-service (QoS) requirement, we have for $j\in\J$:
\begin{equation}
\sum_{\u\in\U_j}c_{j\u}x_{\u}\geq d_{j}.
\label{eq:dj}
\end{equation}

Given $\vec{d}=[d_1,d_2,\ldots,d_m]$, the inequalities system \eqref{eq:sinr}--\eqref{eq:dj} forms a region for $\vec{x}=[x_{1},x_{2},\ldots,x_{|\U|}]$. Within this region, the QoS can be satisfied with the available network resource.  Note that the system is non-linear, as $\rho_k$ appears in the logarithm term in~\eqref{eq:cju}. The user clustering problem is to select a subset of clusters in $\U$ and respectively allocate resource to each selected cluster. For a cluster $\u$ that is not selected, then $x_{\u}=0$. Note that allocating more resource to one cell's cluster may enhance the QoS of the cell, while causing more interference to others. Besides, selecting sub-optimal clusters may lead to over load of cells or failure of meeting the bit demand. Hence the problem is challenging.

\subsection{Comparison to OMA Modeling}
We remark that the models proposed for OMA in~\cite{6732895} are essentially a special case of the NOMA model in this section, i.e. $\U=\{\{1\},\{2\},\ldots,\{m\}\}$. In this case, the intra-cell interference term disappears from~\eqref{eq:sinr}. Since any cluster $\u$ ($\u\in\U$) only contains one UE $j$ ($j\in\J$), the indices ``$j\u$'' (and the index ``$\u$'') can be merged (replaced) to (by) $j$, for \eqref{eq:sinr} and \eqref{eq:cju}-\eqref{eq:dj}. Parameter $x_j$ and $c_j$ then represent respectively the proportion of allocated RB(s) and the achievable capacity for UE $j$ with $j\in\J$. With all these being done, \eqref{eq:cju}--\eqref{eq:dj} can be combined such that $\vec{x}$ is eliminated, leaving $\bm{\rho}$ to be the only variable. This system of cell load $\bm{\rho}$ fulfills the analytical framework of standard interference function (SIF), which enables the computation of the optimal network load settings via fixed-point iterations~\cite{6204009}.

Indeed, by viewing the model as a feasibility problem with variables $\vec{x}$ and $\bm{\rho}$, the orthogonality in OMA enables decomposition among UEs, in terms of the QoS constraints \eqref{eq:dj}. 
The resource allocation is thus on UE-level.
However, in the general NOMA case, one needs to optimize the split of UE demand across multiple clusters containing the same UE. As a result, the clusters sharing UEs couple with each other. The loss of orthogonality therefore leads to a new dimension of complexity in the analysis.

\section{Analytical Results}
\label{sec:opt}

\subsection{Main Results}
In this section, we provide theoretical insights for the proposed model in Section~\ref{sec:sys_model}. The main results are summarized as follows. The model in Section~\ref{sec:sys_model} falls into the framework of SIF with respect to the cell load $\bm{\rho}$. The proof of this conclusion directly leads to a framework for user clustering and power allocation. Algorithms within this framework are able to solve the problem named~\textit{MinF} in~\eqref{eq:minF} ($i\in\I$, $j\in\J$, $\u\in\U$) to global optimum, with any real-valued function $F(\bm{\rho})$ that is monotonically increasing in $\bm{\rho}$.
\begin{equation}
[\textit{MinF}]\quad \min_{\bm{\rho},\vec{p},\vec{x}\geq\vec{0}} F(\bm{\rho})\quad \textnormal{s.t.}~\eqref{eq:sinr}\textnormal{--}\eqref{eq:dj},~\vec{p}\in\P
\label{eq:minF}
\end{equation}

In~\eqref{eq:minF}, $\P$ is a (finite) set of candidate power allocations for a user cluster. In the main analysis, we temporarily fix the power to one of the candidates in $\P$ until Section~\ref{subsec:power}. Note that this simplification is made only for the sake of presentation, without any loss of the generality of our conclusions. In Section~\ref{subsec:power}, we relax this assumption and extend our analytical results to the case with the freedom of power allocation.

\subsection{Single Cell Load Minimization}
\label{subsec:load_opt}
We start with the much simpler problem that concerns a single cell.
Consider any cell $i$ ($i\in\I$). The load minimization problem for cell $i$ is in~\eqref{eq:minload}, with indices $\u\in\U_i$ and $j\in\u$ in \eqref{eq:sinr}\textnormal{--}\eqref{eq:cju} and \eqref{eq:dj}. Variable $\vec{x}_i$ represents the vector of $x_{\u}$ $(\u\in\U_i)$.
\begin{equation}
\min_{\rho_i,\vec{x}_i\geq\vec{0}} \{\rho_i=\sum_{\u\in\U_i}x_{\u}|~\eqref{eq:sinr}\textnormal{--}\eqref{eq:cju},\eqref{eq:dj}\}
\label{eq:minload}	
\end{equation}

Note that in~\eqref{eq:minload}, the loads of all cells other than $i$, i.e., $\rho_{k}$ ($k\in\I\backslash\{i\}$), are treated as parameters instead of variables. With this precondition, the single-cell load minimization is a linear programming (LP) problem and can thus be solved to optimum efficiently. We remark that for any given load of cells $k\in\I\backslash\{i\}$, there is a minimum $\rho_i$. Thus, one can view the minimum $\rho_i$ as a function of $\rho_k$, $k\in\I\backslash\{i\}$. For convenience, we use $\bm{\rho}_{-i}$ to represent the vector of all elements in $\bm{\rho}$ other than $\rho_{i}$. We show in Lemma~\ref{rmk:feasible} below the feasibility of \eqref{eq:minload} for the sake of rigor.

\begin{lemma}
The system of inequalities of \eqref{eq:sinr}\textnormal{--}\eqref{eq:cju},\eqref{eq:dj} is always feasible 
for variables $\rho_i$ and $\vec{x}_i$.
\label{rmk:feasible}
\end{lemma}
The proof of Lemma~\ref{rmk:feasible} is based on Farkas' lemma~\cite{bertsimas1997introduction}. The proof is not shown here due to the limit of space. The problem in~\eqref{eq:minload} can then be defined as a function of $\bm{\rho}_{-i}$, which gives the minimum load for cell $i$, as shown in~\eqref{eq:fi}:
\begin{equation}
f_{i}(\bm{\rho}_{-i}) = \min_{\rho_i,\vec{x}\geq\vec{0}} \{\rho_i=\sum_{\u\in\U_i}x_{\u}|~\eqref{eq:sinr}\textnormal{--}\eqref{eq:cju},\eqref{eq:dj}\}.
\label{eq:fi}
\end{equation}
\begin{lemma}
No infinite discontinuity exists for $f_i(\bm{\rho}_{-i})$.
\label{rmk:domain}
\end{lemma}
Lemma~\ref{rmk:domain} states that $f_i(\bm{\rho}_{-i})$ is real-valued in its domain. The lemma is induced by Lemma~\ref{rmk:feasible} and that $\rho_i=0$ is a lower bound of~\eqref{eq:minload}, such that 
the optimal objective in the LP cannot be $-\infty$.

\subsection{Standard Interference Function}
\label{subsec:analysis}

Network-wise, we have the function $\vec{f}(\bm{\rho})$ defined element-wisely in~\eqref{eq:fi2} for $\I$. By Theorem~\ref{thm:sif}, $\vec{f}(\bm{\rho})$ is an SIF.
\begin{equation}
\vec{f}(\bm{\rho})=[f_1(\bm{\rho}_{-1}),f_2(\bm{\rho}_{-2}),\ldots,f_n(\bm{\rho}_{-n}))]	
\label{eq:fi2}
\end{equation}
\begin{theorem}
$\vec{f}(\bm{\rho})$ is an SIF, i.e.
the
following properties hold: 
\begin{enumerate} 
\item (Scalability)
$\alpha\vec{f}(\bm{\rho})>\vec{f}(\alpha\bm{\rho}),~
\bm{\rho}\in \mathbb{R}^n_+,~\alpha>1$.
\item (Monotonicity) $\vec{f}(\bm{\rho})\geq \vec{f}(\bm{\rho}')$,
$\bm{\rho} \geq \bm{\rho}'$, $\bm{\rho},\bm{\rho}'\in \mathbb{R}^n_+$.
\end{enumerate} 
\label{thm:sif}
\end{theorem}

The proof of Theorem~\ref{thm:sif} is in the Appendix. Any function satisfying the two properties in Theorem~\ref{thm:sif} falls into the category of SIF. We explain the main properties of SIF as follows. For the non-linear equation system $\vec{f}(\bm{\rho})=\bm{\rho}$, if there exists a feasible solution $\bm{\rho}^{*}\in\mathbb{R}^{n}_{+}$, i.e., equation $\vec{f}(\bm{\rho}^{*})=\bm{\rho}^{*}$ holds, then $\bm{\rho}^{*}$ (named as the fixed point of $\vec{f}(\bm{\rho})$) uniquely exists. Another property is that, $\bm{\rho}^{*}$ can be computed by fixed-point iterations, iteratively by the equation $\bm{\rho}^{(k)}=\vec{f}(\bm{\rho}^{(k-1)})$ with $k\geq 1$ and any $\bm{\rho}^{(0)}\in\mathbb{R}_{+}^{n}$. With the existence of $\bm{\rho}^{*}$, starting from any $\bm{\rho}^{(0)}\in\mathbb{R}_{+}^{n}$, the iterations eventually converge to $\bm{\rho}^{*}$. Denote by $\vec{f}^{k}$ ($k> 1$) the function composition of $\vec{f}(\vec{f}^{k-1}(\bm{\rho}))$. We formally state this property in Lemma~\ref{rmk:lim}.
\begin{lemma}
If $\lim_{k\rightarrow\infty}\vec{f}^{k}(\bm{\rho})$ exists for any $\bm{\rho}\in\mathbb{R}^{n}_{+}$, it exists uniquely for all $\bm{\rho}\in\mathbb{R}^{n}_{+}$ and is independent of $\bm{\rho}$.
\label{rmk:lim}
\end{lemma}



\subsection{User Clustering}
\label{subsec:nopower}
\textit{MinF} with fixed power allocation is essentially a user clustering problem.
Based on the analysis in Section~\ref{subsec:analysis}, we derive sufficient and necessary conditions for \textit{MinF} with fixed power allocation, in terms of its feasibility and optimality, in Theorem~\ref{thm:feasibility} and Theorem~\ref{thm:optimality}, respectively. The proofs of both theorems are detailed in the Appendix due to their rather technical nature. Note that though the variables in \textit{MinF} (with fixed power) are $\bm{\rho}$ and $\vec{x}$, the conditions shown in the two theorems only concern $\bm{\rho}$. This is because, when evaluating the function $\vec{f}(\bm{\rho})$, $\vec{x}$ is accordingly computed by solving corresponding LPs in \eqref{eq:minload}. Thus we omit $\vec{x}$ in our following discussion for the sake of presentation.

\begin{theorem}
For fixed-power \textit{MinF},
$\bm{\rho}$ ($\bm{\rho}\leq\bar{\rho}\vec{1}$) is feasible if
and only if the load $\vec{f}(\bm{\rho})$ is feasible and
$\bm{\rho}\geq \vec{f}(\bm{\rho})$.
\label{thm:feasibility}
\end{theorem}
Besides feasibility, for problem solving, 
Theorem~\ref{thm:feasibility} provides an efficient and effective method for improving any feasible solution to \textit{MinF}. For any feasible solution $\bm{\rho}$, $\vec{f}(\bm{\rho})$ yields a better one\footnote{Rigorously speaking, the new solution is only guaranteed to be no worse by Theorem~\ref{thm:feasibility}. However in fact it is guaranteed to be better unless the old one is already at the optimum. A proof can be easily derived based on Lemma~\ref{rmk:fp}.}. One can compute $\vec{f}(\bm{\rho})$ and use it to replace $\bm{\rho}$ as a better solution for \textit{MinF}, by solving $n$ LP problems.


\begin{theorem} 
Load $\bm{\rho}^{*}$ is at the optimum of fixed-power 
\textit{MinF} if and only if $\bm{\rho}^{*} = \vec{f}(\bm{\rho}^{*})\leq\bar{\rho}\vec{1}$.
\label{thm:optimality}
\end{theorem}

Theorem~\ref{thm:optimality} shows that, the optimal solution of \textit{MinF} is at the fixed point of the function $\vec{f}(\bm{\rho})$. In addition, Theorem~\ref{thm:optimality} reveals the relationship between the feasibility of the proposed model and the function $\vec{f}(\bm{\rho})$. Suppose that we have a feasible solution $\bm{\rho}$ for \textit{MinF}. Since \textit{MinF} is bounded below by $F(\vec{0})$, we conclude the existence of the optimum of \textit{MinF}, which, by Theorem~\ref{thm:optimality}, resulting in the existence of the fixed point for $\vec{f}(\bm{\rho})$. Hence the existence of the fixed point for $\vec{f}(\bm{\rho})$ is necessary for the feasibility of \textit{MinF}. Also, if the fixed point of $\vec{f}(\bm{\rho})$ exists (i.e. there is a $\bm{\rho}^{*}$ such that $\bm{\rho}^{*} = \vec{f}(\bm{\rho}^{*})$ holds), then $\bm{\rho}^{*}$ is an optimal solution to \textit{MinF}, showing $\bm{\rho}^{*}$ feasibility. Therefore the existence of the fixed point for $\vec{f}(\bm{\rho})$ is sufficient for the feasibility of \textit{MinF} as well. The conclusion is summarized in Lemma~\ref{rmk:fp} below.

\begin{lemma}
The fixed-power \textit{MinF} is feasible if and only if the fixed point exists for $\vec{f}(\bm{\rho})$.
\label{rmk:fp}	
\end{lemma}

Starting from any $\bm{\rho}^{(0)}\in\mathbb{R}^{+}_{n}$, we run the fixed point iterations $\bm{\rho}^{(k)}=\vec{f}(\bm{\rho}^{(k-1)})$ for $k\geq 1$. During each iteration, $n$ LPs in~\eqref{eq:minload} for $i\in\I$ are respectively solved. At the convergence, the optimum is reached. One may also strike a balance between the optimality and the computational efficiency. Once we know that $\bm{\rho}^{(k)}$ is feasible for any $k\geq 0$, then all $\bm{\rho}^{(k+1)},\bm{\rho}^{(k+2)},\ldots$ are feasible as well, and one can terminate the iterations at any step after $k$, to obtain a sub-optimal solution. 
\begin{figure}[!ht]
\centering
  \includegraphics[width=0.9\linewidth]{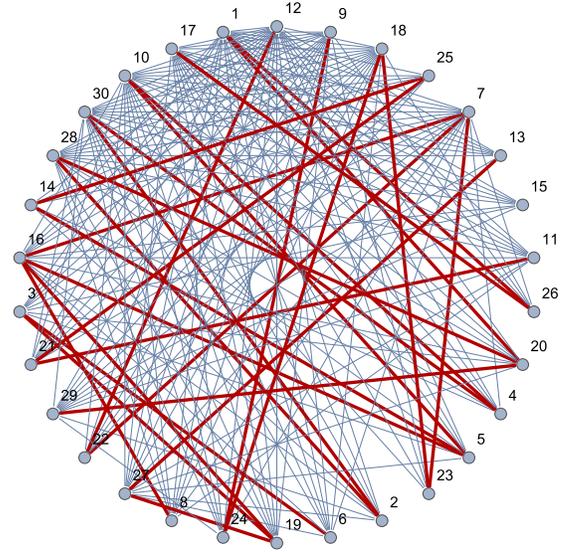}
  \caption{This figure comes from one of our simulations and it is used as an illustration for one cell's user clustering in multi-cell scenarios. There are $30$ UEs within this cell. Each vertex represents a UE and each edge a candidate clustering option. Starting from UE $12$, all UEs are sorted decreasingly according to its power gain from the cell and arranged clockwise (i.e. UE $12$ has the best channel condition and UE $1$ has the worst). The highlighted $28$ edges are selected among all $181$ candidate ones  by~\eqref{eq:x} and~\eqref{eq:Ustar} via solving LPs. Note that not all the UEs are expected to use NOMA, e.g., UE $15$ is not clustered with others. From the visualization, it rarely happens that one UE is clustered with another with similar channel condition (e.g. its near neighbors in the circle).}
 \label{fig:Ustar}
\end{figure}

 Mathematically, the corresponding $\vec{x}^{*}$ for $\bm{\rho}^{*}$ is formulated in~\eqref{eq:x} with $i\in\I$.
\begin{equation}
	\vec{x}^{*}_i = \argmin_{\vec{x}_i} f_i(\bm{\rho}^{*})
\label{eq:x}
\end{equation}
Accordingly we obtain the optimal user clustering solution, denoted by $\U^{*}$ ($\U^{*}\subseteq\U$), in~\eqref{eq:Ustar}. 
\begin{equation}
\U^{*} = \{\u~|~x^{*}_{\u}>0, \u\in\U\}
\label{eq:Ustar}
\end{equation}
\figurename~\ref{fig:Ustar} gives an illustration of user clustering. Note that not all the UEs are expected to use NOMA, and 
the clustering occurs between UEs with large variation in channel conditions.

For a cell $i$ ($i\in\I$), given the information of other cells' load $\bm{\rho}_{-i}$, solving $f_i(\bm{\rho})$ is based on local information, making it suitable to be run in a distributed manner. The technique called ``asynchronous fixed-point iterations''~\cite{414651} can be used. It means that for an arbitrary subset $\I^{\text{sub}}$ ($\I^{\text{sub}}\subseteq\I$)
 one can do fixed-point iterations for $\vec{f}(\bm{\rho})$ by following the rules that 1) $\rho_i^{(k)}=f_i(\bm{\rho}^{(k-1)})$ for any $i\in\I^{\text{sub}}$ and 2) $\rho_i^{(k)}=\rho_i^{(k-1)}$ for any $i\in\I\backslash\I^{\text{sub}}$, and $k\geq 1$, without loss of the convergence property. The solution obtained by such an iterative process still possesses feasibility for \textit{MinF}, as can be verified by Theorem~\ref{thm:feasibility}. The asynchronous fixed-point iterations converge to the fixed point of $\vec{f}(\bm{\rho})$ and optimality holds as well~\cite{414651}\footnote{An intuitive explanation is that, the fixed point is unique, regardless of how we reach it.}. Therefore, it is sufficient for a cell to have information of a subset of cells (e.g., the surrounding cells) having major significance in terms of interference. The update for such information is very local and hence easily implemented via the LTE X2 interface. 
 
\subsection{Solving MinF}
\label{subsec:power}
All the conclusions in Section~\ref{sec:opt} still hold with power allocation taken into consideration. An intuitive explanation is provided as follows. First, note that one can decompose the power allocation in terms of cells, as it can be seen in~\eqref{eq:sinr} that, any cell $i$ interferes with other cells with $p_i$, independent of the power splitting scheme being used in cell $i$. Consider any cell $i\in\I$. Suppose there are in total $K_i$ allocation schemes for cell $i$. 
Respectively, denote by 
$f^{[1]}_i(\bm{\rho}_{-i}),f^{[2]}_i(\bm{\rho}_{-i}),\ldots,f^{[K_i]}_i(\bm{\rho}_{-i})$ 
the function $f_i(\bm{\rho}_{-i})$ in Section~\ref{subsec:analysis} under 
the $K_i$ candidate power allocations. The load optimization problem in Section~\ref{subsec:load_opt} evolves to \eqref{eq:fi'}.
\begin{equation}
f'_i(\bm{\rho}_{-i})=\min\left\{f^{[1]}_i(\bm{\rho}_{-i}),f^{[2]}_i(\bm{\rho}_{-i}),\ldots,f^{[K_i]}_i(\bm{\rho}_{-i})\right\}
\label{eq:fi'}
\end{equation}

The function in~\eqref{eq:fi'} is also SIF, because both the scalability and the monotonicity hold for $f'_i(\bm{\rho}_{-i})$. We denote by $\vec{f}'(\bm{\rho})$ the vector version of~\eqref{eq:fi'} with $i\in\I$. As the SIF properties hold for the new function $\vec{f}'(\bm{\rho})$, all the conclusions in Section~\ref{subsec:nopower} naturally remain valid for $\vec{f}'(\bm{\rho})$, accordingly with the notation $\vec{f}$ (or $f_i$) changed to $\vec{f}'$ (or $f'_i$) in all the theorems' statements as well as in their corresponding proofs, and the word ``fixed-power'' in all the theorems' statements can thus be removed. 
Note that for evaluating the expression of $\vec{f}'(\bm{\rho})$, one needs to solve $\sum_{i=1}^{n}K_{i}$ instead of $n$ LP problems as for $\vec{f}(\bm{\rho})$. 

Denote by $\bm{\rho}'^{*}$ the fixed point of $\vec{f}'(\bm{\rho})$, i.e. $\bm{\rho}'^{*}=\vec{f}'(\bm{\rho}'^{*})$. We use $\vec{p}_i$ to represent the vector of $p_{j\u}$ with $u\in\U_i$ and $j\in\u$, and $\P_i$ the candidate set of power allocation schemes for cell $i$ (i.e. the $K_i$ schemes as mentioned above). The optimal power allocation $\vec{p}^{*}_i$ for \textit{MinF} is given by~\eqref{eq:power_allocation}, for $i\in\I$. In other words, $\vec{p}^{*}_i$ corresponds to the $k_{th}$ power allocation scheme, which leads to the minimum among all the $K_i$ functions $f_{i}^{[k]}(\bm{\rho}'^{*}_{-i})$ ($k\in[1,K_i]$) in~\eqref{eq:fi'} at the convergence. 
\begin{equation}
\vec{p}_i^{*} = \argmin_{\vec{p}_i\in\P_i} f'_{i}(\bm{\rho}'^{*})	
\label{eq:power_allocation}
\end{equation}
The optimal clustering with power allocation $\vec{p}^{*}$ can be obtained by using~\eqref{eq:x} and~\eqref{eq:Ustar} with $f_i$ replaced by $f'_i$ and $\bm{\rho}^{*}$ replaced by $\bm{\rho}'^{*}$ respectively. 
\section{Numerical Results}
\label{sec:numerical}

\subsection{Simulation Settings}

We consider three performance metrics, the total load $\norm{\bm{\rho}}_1$, the maximum load $\norm{\bm{\rho}}_{\infty}$, and the efficiency of achieved rate on RBs. The rate efficiency is defined as the ratio between the sum of all user demands and the total of consumed RBs. 
We consider heterogeneous network scenarios in the simulation. Six small cells (SCs) are deployed around one macro cell (MC). The parameter setting is in~\tablename~\ref{tab:sim}. 
\begin{table}[!h]
\centering
\caption{Simulation Parameters.}
\begin{tabular}{ll}
\toprule
\textbf{Parameter} & \textbf{Value} \\
Cell radius & 500 m\\
Carrier frequency & 2 GHz \\
Total bandwidth & 20 MHz\\
Cell coverage radius & MC: 500 m; SC: 100 m \\
Number of users & $\{70,140,210,280,350\}$ \\
Cell load limit $\bar{\rho}$ & $\{0.4,0.6,0.8,1.0\}$ \\
Path loss & COST-231-HATA \\
Shadowing (Log-normal) & MC: 8 dB standard deviation\\
					   & SC: 4 dB standard deviation\\
Fading & Rayleigh flat fading \\
Noise power spectral density & -173 dBm/Hz \\
Total power on RB & MC: 800 mW \\
				  & SC: 100 mW\\
$\alpha_{\textnormal{FTPC}}$ & $\{0.2,0.4,0.6,0.8\}$ \\
$\alpha_{\textnormal{NTT}}$ & $\{0.1,0.2,0.3,0.4\}$ \\
\bottomrule
\end{tabular}
\label{tab:sim}	
\end{table}

The UE demands are set in correspondence to the value of $\bar{\rho}$, such that the load of at least one cell in OMA reaches the limit $\bar{\rho}$. Two power allocation schemes are used for performance comparison. The ``fractional transmit power control'' (FTPC) proposed in~\cite{6666209} uses a parameter $\alpha_{\textnormal{FTPC}}\in[0,1]$ to control the fairness for power splitting among UEs. In \cite{R1-153332}, power allocation based on a pre-determined power ratio set is suggested, with a proportion $\alpha_{\textnormal{NTT}}\in(0,0.5)$ for allocating power to the UE with better channel condition. We use ``NTT'' to represent this power allocation scheme. Two sets of candidate parameter values of $\alpha_{\textnormal{FTPC}}$ and $\alpha_{\textnormal{NTT}}$ in \tablename~\ref{tab:sim} are used respectively for the two power allocation schemes, in computing~\eqref{eq:power_allocation}. Uniform power allocation (i.e. two UEs in NOMA are allocated with the same amount of power), referred to as ``Uniform'', is used as reference. OMA is used as the baseline for performance benchmarking. The other parameters in~\tablename~\ref{tab:sim} are coherent with~\cite{tr36913}. 

\subsection{Performance Evaluation}
In summary, the numerical results show significant improvement by NOMA on resource and rate efficiency. Power allocation plays an important role in enhancing the performance in multi-cell NOMA. NOMA is promising in the scenario with intensive data traffic and high user densities.

In \figurename~\ref{fig:sumload}, with higher demand, the reduction on total load achieved by NOMA becomes larger, meaning that NOMA is preferred in the scenarios with high traffic density. There is no difference between FTPC and NTT. On the other hand, both are considerably better than Uniform, meaning that power allocation in multi-cell NOMA has significant influence on resource efficiency.

\begin{figure}[!h]
\centering
  \includegraphics[width=\linewidth]{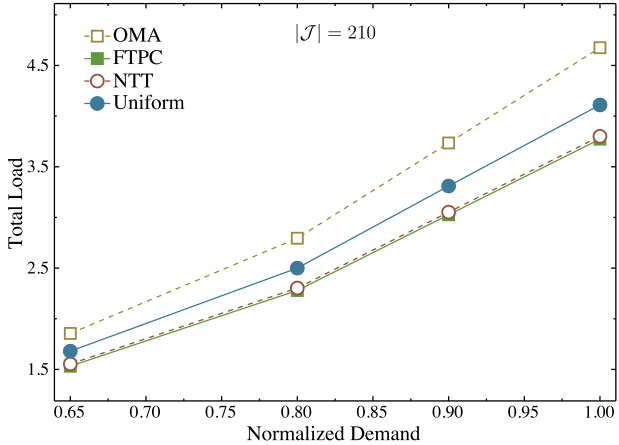}
  \caption{Performance of the total load of the network. The objective is $F(\bm{\rho})=\norm{\bm{\rho}}_1$. The number of UEs is $210$.}
 \label{fig:sumload}
\end{figure}
\begin{figure}[!h]
\centering
  \includegraphics[width=\linewidth]{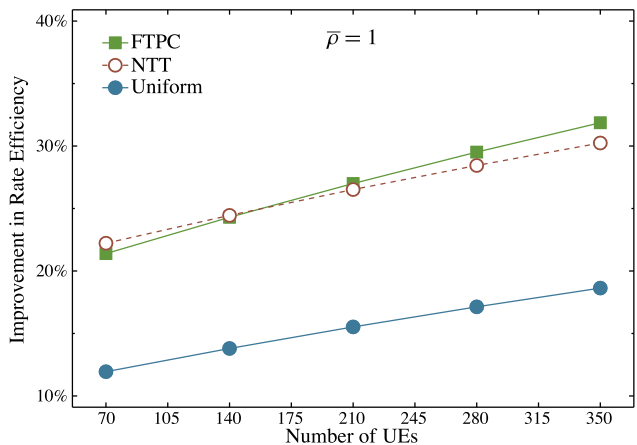}
  \caption{Performance of the improvement in rate efficiency. OMA is the baseline. The load limit $\bar{\rho}$ equals $1.0$, meaning that for every data point, at least one cell in OMA is at full load.}
\label{fig:avgrate}
\end{figure}

\figurename~\ref{fig:avgrate} shows the rate efficiency improvement with respect to the network density, with OMA being the baseline. The parameter $\bar{\rho}$ is set to $1.0$ such that at least one cell in OMA is in full load. With the increase of the network density, NOMA achieves larger improvement in rate efficiency. Compared to Uniform, both FTPC and NTT lead to better performance, and FTPC has slight advantage over NTT when the user density is large. On the contrary, NTT leads to slightly better performance with low user density. The difference on the rate performance between Uniform and the other two becomes higher with the increase of the network density.
\begin{table}[!h]
\centering
  \caption{Performance Evaluation with $F(\bm{\rho})=\norm{\bm{\rho}}_{\infty}$.}
  \begin{tabular}{lcc}
  \toprule
    \textbf{Scheme} & \textbf{Load Reduction} & \textbf{Improvement in Rate Efficiency} \\
    FTPC & 19.9\% & 25.2\%  \\
    NTT & 19.4\% & 24.2\% \\
    Uniform & 11.7\% & 13.0\% \\
   \bottomrule
  \end{tabular}
\label{tab:maxload}
\end{table}

For $F(\bm{\rho})=\norm{\bm{\rho}}_{\infty}$, we minimize the load for the most heavy-loaded cell in the network, and evaluate the performance in terms of its load reduction and rate efficiency improvement. The settings of demands and the number of UEs follow those in \figurename~\ref{fig:sumload} and \figurename~\ref{fig:avgrate}. By using OMA as the baseline, the numerical results of improvement are averaged and summarized in \tablename~\ref{tab:maxload}, which is coherent with the results in \figurename~\ref{fig:sumload} and \figurename~\ref{fig:avgrate}.

%
\section{Conclusion}
\label{sec:conclusion}
In this paper, multi-cell NOMA has been put into an optimization framework. We conclude that NOMA is a promising technique for raising spectrum efficiency, especially in scenarios with intensive data traffic and high user densities. 



\section{Acknowledgement}
This work has been supported in part by the Swedish Research Council and the Luxembourg National
Research Fund (FNR) Multi-Annual Thematic Research Programme (CORE) projects SEMIGOD.

%
%


\bibliographystyle{IEEEtran}
\bibliography{ref.bib}

\appendix
 
\section*{\small Proof of Lemma~\ref{rmk:decoding}}
\begin{proof}
The two UEs $j$ and $h$ are in the cluster $\u$ and served by cell $i$. Denote by $\gamma_{hj}$ and $\gamma_{hh}$ the SINR at user $j$
and $h$, in respect of the transmission for user $h$, in~\eqref{eq:hj} and \eqref{eq:hh}.
\begin{equation}
    \gamma_{hj}=\frac{p_{h\u}g_{ij}}{p_{j\u}g_{ij}+
        \sum_{k\in\I\backslash\{i\}}
  p_{k}g_{kj}\rho_k
+\sigma^2}
\label{eq:hj}
\end{equation}
\begin{equation}
    \gamma_{hh}=\frac{p_{h\u}g_{ih}}{p_{j\u}g_{ih}+
    \sum_{k\in\I\backslash\{i\}}
  p_{k}g_{kj}\rho_k+\sigma^2}
  \label{eq:hh}
\end{equation}

The condition for UE $j$ to decode UE $h$ is $\gamma_{hj}\geq\gamma_{hh}$, i.e.
\begin{multline}
    \gamma_{hj}\geq\gamma_{hh}\Leftrightarrow
p_{j\u}g_{ij}g_{ih}+g_{ij}\sum_{k\in\I\backslash\{i\}}
    p_{k}g_{kh}\rho_{k}+g_{ij}\sigma^2 \\
    \geq p_{j\u}g_{ij}g_{ih}+g_{ih}\sum_{k\in\I\backslash\{i\}}p_{k}g_{kj}\rho_k+g_{ih}\sigma^2 \\
\Leftrightarrow
\sum_{k\in\I\backslash\{i\}}
p_{k}\rho_{k}(g_{ih}g_{kj}-g_{ij}g_{kh})
\leq (g_{ij}-g_{ih})\sigma^2
\label{eq:decoding}
\end{multline}
Recall that $g_{ij}>g_{ih}$. Thus the right-hand side of~\eqref{eq:decoding} is positive. By Lemma~\ref{rmk:decoding} that $g_{ij}/g_{ih}\geq g_{kj}/g_{kh}$ for $k\in\I\backslash\{i\}$, the left-hand side is negative. Hence \eqref{eq:decoding} holds.
\end{proof}

\section*{\small Proof of Theorem~\ref{thm:sif}}
We reformulate the problem in~\eqref{eq:minload} below, for the sake of clarity of the proof.
\begin{subequations}
\begin{alignat}{2}
\quad &
\min\limits_{\rho_i,\vec{x}_i,\vec{r}\geq\vec{0}} \rho_i \\
s.t. \quad &     c_{j\u}=MB
    \log\left(1+\gamma_{j\u}(\bm{\rho}_{-i})
\right)
\label{eq:cju2} \\
& \rho_{i} = \sum_{\u\in\U_i}x_{\u} 
\label{eq:rhoi2} \\
&\sum_{u\in\U_j} r_{j\u} \geq d_j	
\label{eq:rju2}\\
& x_{\u}\geq\frac{r_{j\u}}{c_{j\u}(\bm{\rho}_{-i})}
\label{eq:xu2} 
\end{alignat}
\label{eq:reformulate}
\end{subequations}
\begin{proof}
(Monotonicity) Suppose $\bm{\rho'}\leq\bm{\rho}$. We change $c_{j\u}(\bm{\rho})$
to $c_{j\u}(\bm{\rho}')$. According to the monotonicity of the function
$c_{j\u}$, we have $c_{j\u}(\bm{\rho})\leq c_{j\u}(\bm{\rho}')$, for any
$u\in\U$ and $j\in\u$. Note that any feasible solution $(\bm{x},\bm{r})$ for the
minimization problem with $\bm{\rho}$ is still feasible for the minimization
problem with $\bm{\rho}'$. Therefore, the minimization problem in \eqref{eq:fi2}
is relaxed with $\bm{\rho}$ being replaced by $\bm{\rho}'$. Therefore, we have
$f_i(\bm{\rho}')\leq f_i(\bm{\rho})$.

(Scalability) First, note that the equality $\alpha f_i(\bm{\rho}) =
\min_{\vec{x},\vec{r}\geq\vec{0}}
\{\alpha\rho_i|~\eqref{eq:cju2}\textnormal{--}\eqref{eq:rju2},r_{j\u}\leq
x_{\u}c_{j\u}(\bm{\rho})\}$, for which we denote the optimal solution by $(\vec{x}',\vec{r}')$. Consider the problem $\beta$, i.e., 
$\beta:\min_{\vec{x},\vec{r}\geq\vec{0}}
\{\rho_i|~\eqref{eq:cju2}\textnormal{--}\eqref{eq:rju2},x_{\u}\geq \alpha
r_{j\u}/c_{j\u}(\bm{\rho})\}$ with $\alpha>1$. One can verify that $(\alpha\vec{x}',\vec{r}')$ is a feasible solution to the problem $\beta$, with objective value $\alpha f_i(\bm{\rho})$. Then, the optimum of problem $\beta$ is no more than $\alpha f_i(\bm{\rho})$. Suppose we replace $\bm{\rho}$ with
$\alpha\bm{\rho}$ ($\alpha>1$) in the minimization problem~\eqref{eq:fi2}. By the
scalability of $1/c_{j\u}(\bm{\rho})$, we have $1/c_{j\u}(\alpha\bm{\rho})<
\alpha/c_{j\u}(\bm{\rho})$. Thus, the minimization problem corresponding to $f_i(\alpha\bm{\rho})$ is a relaxation of the problem $\beta$. 
For the relaxed minimization problem, the optimal objective value $f_{i}(\alpha\bm{\rho})$ is less than that of $\beta$.
Therefore, we conclude $f_{i}(\alpha\bm{\rho})<\alpha f_i(\bm{\rho})$. Hence the conclusion. 
\end{proof}

\section*{\small Proof of Theorem~\ref{thm:feasibility}}
\begin{lemma}
For any $\bm{\rho}\geq\vec{0}$, if there exists $i\in\I$ such that
$\rho_i<f_i(\bm{\rho}_{-i})$, then $\bm{\rho}$ is not feasible to
\eqref{eq:sinr}\textnormal{--}\eqref{eq:dj}.
\label{lma:not_feasible}
\end{lemma}
\begin{proof}
Let $\rho'_i = f_i(\bm{\rho}_{-i})$. By the definition of $f_i$, $\rho'_i$ is
the minimum value satisfying
\eqref{eq:sinr}\textnormal{--}\eqref{eq:cju}, and \eqref{eq:dj} under
$\bm{\rho}_{-i}$. Therefore any $\rho_i$ with $\rho_i<\rho'_i$ causes at least one of the
constraints~\eqref{eq:sinr}\textnormal{--}\eqref{eq:cju}, or \eqref{eq:dj} being
violated with $\bm{\rho}_{-i}$, meaning that the vector $\bm{\rho}$ cannot satisfy all constraints \eqref{eq:sinr}\textnormal{--}\eqref{eq:dj}. Hence the conclusion.
\end{proof}
\begin{proof}
Theorem~\ref{thm:feasibility} is proved as follows. By the inverse proposition of
Lemma~\ref{lma:not_feasible}, a feasible solution $\bm{\rho}$ always satisfies $\bm{\rho}\geq\vec{f}(\bm{\rho})$. Now suppose $\bm{\rho}$ is feasible to \textit{MinF} and consider using $\vec{f}(\bm{\rho})$ as a solution to \textit{MinF}. (Together with the $\vec{x}$ obtained when computing $\vec{f}(\bm{\rho})$.) 
Then $\vec{f}(\bm{\rho})$ satisfies~\eqref{eq:rhoi}. Also $\vec{f}(\bm{\rho})$ together with its $\vec{x}$ fulfills~\eqref{eq:sinr}\textnormal{--}\eqref{eq:cju} and \eqref{eq:dj} by the definition of $\vec{f}(\bm{\rho})$. Thus 
$\vec{f}(\bm{\rho})$ is feasible.

%
%
%

For the sufficiency, note that the feasibility of $\vec{f}(\bm{\rho})$ indicates that $\bm{\rho}_{-i}$ along with $\vec{x}_i$ obtained by solving $f_i(\bm{\rho}_{-i})$ satisfies \eqref{eq:sinr}\textnormal{--}\eqref{eq:cju}, and \eqref{eq:dj}.  Combined with the precondition $\rho_i\leq\bar{\rho}$ with $i\in\I$, the load $\bm{\rho}$ is feasible to \eqref{eq:sinr}\textnormal{--}\eqref{eq:dj} (and thus feasible to \textit{MinF}). Hence the conclusion.
\end{proof}

\section*{\small Proof of Theorem~\ref{thm:optimality}}
\begin{proof}

(Necessity) If $\bm{\rho}^{*}$ is feasible, then obviously we have $\bm{\rho}^{*}\leq\bar{\rho}\vec{1}$.
By
Theorem~\ref{thm:feasibility}, $\vec{f}(\bm{\rho}^{*})$ 
is also feasible and 
$\vec{f}(\bm{\rho}^{*})\leq\bm{\rho}^{*}$ holds. Also, $\vec{f}^{k}(\bm{\rho}^{*})$ for any $k\geq 1$ is a 
feasible solution. According to Theorem~\ref{thm:sif}, $\vec{f}(\bm{\rho})$ is monotonic in $\bm{\rho}$, and thus we have $\vec{f}^{k}(\bm{\rho}^{*})\geq\vec{f}^{k+1}(\bm{\rho}^{*})$ for any $k\geq 1$. Based on Lemma~\ref{rmk:lim}, we let $\bm{\rho}'=\lim_{k\rightarrow\infty}\vec{f}^{k}(\bm{\rho}^{*})$.
Then $\bm{\rho}'\leq \bm{\rho}^{*}$ holds, by the above discussion. In addition, note that $\bm{\rho}'$ is a feasible solution as well. By that $\bm{\rho}^{*}$ is optimal for \textit{MinF}, we have $\bm{\rho}'=\bm{\rho}^{*}$, otherwise $\bm{\rho}'$ would lead to a better objective value in \textit{MinF} than $\bm{\rho}^{*}$. Hence $\bm{\rho}^{*}=\lim_{k\rightarrow\infty}\vec{f}^{k}(\bm{\rho}^{*})$, i.e. $\bm{\rho}^{*}=\vec{f}(\bm{\rho}^{*})$.

(Sufficiency) By Theorem~\ref{thm:feasibility}, for any feasible $\bm{\rho}$, $\lim_{k\rightarrow\infty}\vec{f}^{k}(\bm{\rho})$ is feasible and $\lim_{k\rightarrow\infty}\vec{f}^{k}(\bm{\rho})\leq\bm{\rho}$ holds. By Lemma~\ref{rmk:lim}, the limit is unique for any $\bm{\rho}\geq\vec{0}$, and thus $\lim_{k\rightarrow\infty}\vec{f}^{k}(\bm{\rho})=\lim_{k\rightarrow\infty}\vec{f}^{k}(\bm{\rho}^{*})$. Since $\bm{\rho}^{*}=\vec{f}(\bm{\rho}^{*})$, we have $\bm{\rho}^{*}=\lim_{k\rightarrow\infty}\vec{f}^{k}(\bm{\rho}^{*})$. Thus $\bm{\rho}^{*}\leq\bm{\rho}$ for any feasible $\bm{\rho}$, meaning that $\bm{\rho}^{*}$ is optimal for \textit{MinF}.
Hence the conclusion.
\end{proof}

\end{document}